\documentclass{amsart}
\usepackage{mathrsfs} 

\title{On the Skew Stickiness Ratio}
\author{Masaaki Fukasawa}
\address{The University of Osaka, 560-8531 Japan}
\date{}

\newtheorem{lemma}{Lemma}
\newtheorem{proposition}{Proposition}
\newtheorem{theorem}{Theorem}

\newtheorem{remark}{Remark}
\newtheorem{hypothesis}{Hypothesis}

\begin{document}

\begin{abstract}
  The skew stickiness ratio is a statistic that captures the joint dynamics of an asset price and its volatility. We derive a representation formula for this quantity using the It\^o–Wentzell and Clark–Ocone formulae, and we apply it to analyze its  asymptotics under Bergomi‑type stochastic volatility models.
\end{abstract}

\maketitle

\section{Introduction}

The Skew Stickiness Ratio (SSR), introduced by Bergomi~\cite{Bergomi}, serves as a quantitative measure of how different models generate distinct implied volatility dynamics. Conventional stochastic volatility models yield SSR values that diverge from empirical observations. Since SSR captures the cross gamma risk associated with stochastic volatility, employing a model that aligns with market-consistent SSR values is crucial for effective derivative hedging.

Let $S = \{S_t\}$ be a positive continuous local martingale standing for an asset price process.
For a fixed maturity $T>0$, a strike price $K>0$, 
the put option price process { is denoted by $P(K) = \{P_t(K)\}$, and
the implied total variance process $\Sigma(K) = \{\Sigma_t(K)\}$ is defined through
\begin{equation}\label{def1}
P_t(K) = p_K(S_t,\Sigma_t(K)),
\end{equation}
where} 
$p_K$ is the Black-Scholes put option price given by
\begin{equation*}
    p_K(s,{ w}) = K\Phi(-d_-)-s\Phi(-d_+), \ \ d_\pm = d_\pm(s,{ w}) = \frac{\log \frac{s}{K} \pm \frac{1}{2}{ w}}{\sqrt{{ w}}},
\end{equation*}
and $\Phi$ is the standard normal distribution function.
The implied volatility $\sigma(K) = \{\sigma_t(K)\}$, 
the at-the-money volatility
$\sigma^S = \{\sigma^S_t\}$,
the at-the-money skew
$\sigma^\prime = \{\sigma^\prime_t\}$, and then the SSR process $R = \{R_t\}$ are defined through
\begin{equation}\label{def2}
\sigma_t(K) = \sqrt{\frac{\Sigma_t(K)}{T-t}}, \ \ 
\sigma^S_t = \sigma_t(S_t), \ \ 
      \sigma^\prime_t = \frac{\mathrm{d}}{\mathrm{d}k}\Big|_{k=0} \sigma_t(S_te^k), \ \ 
   R_t  = \frac{1}{\sigma^\prime_t}
    \frac{\mathrm{d}\langle \sigma^S,\log S \rangle_t}{\mathrm{d}\langle \log S\rangle_t}
\end{equation}
respectively. 

The empirical version of SSR is computed as the regression coefficient of the increments of the market-implied at-the-money volatility with respect to the underlying log returns, normalized by the market-implied at-the-money skew.
According to Bergomi~\cite{Bergomi}, the empirical SSR process took values around $3/2$ for time-to-maturities $T-t$ ranging from a month to a few years in 
 the Euro Stoxx 50 and S\&P 500 markets, while Bourgeys et al.~\cite{BDD} show that the values have become smaller in recent years. The empirical SSR can have different dynamics in other index markets, as noted also in  Bergomi~\cite{Bergomi}.

While the at-the-money volatility is observable in markets and so the empirical SSR is a convenient statistic, the covariation $\langle \sigma^S,\log S \rangle$ under a given model has not been rigorously computed in the literature. 
In fact, Bergomi~\cite{Bergomi} approximated $R$ by formally replacing $\sigma^S$ by the square root of the averaged forward variance 
\begin{equation*}
    \sigma^{V}_t = \sqrt{\frac{1}{T-t}\int_t^T V_t(s)\mathrm{d}s},\ \ V_t(s) = \mathsf{E}[V_s|\mathscr{F}_t]
\end{equation*}
without any rigorous estimate of the approximation error.
Bourgey et al.~\cite{BDD} distinguished between 
 this forward variance version and the original SSR, with theoretical analysis on the former while providing only a formal approximation and numerical evidences for the latter.
Friz and Gatheral~\cite{FG} treated directly
$\langle \sigma^S,\log S \rangle$ 
while relying on a formal application of It\^o's formula with a functional derivative
\begin{equation*}
     \frac{\mathrm{d}\langle \sigma^S,\log S \rangle_t}{\mathrm{d}\langle \log S\rangle_t}= \int_t^T \mathrm{d}u \frac{\delta \sigma^S}{\delta V_t(u)}  \frac{\mathrm{d}\langle V(u),\log S \rangle_t}{\mathrm{d}\langle \log S\rangle_t}
\end{equation*}
for which a verification seems not trivial.
Bourgey et al.~\cite{BDD2} also studied the original SSR relying on an infinite-dimensional It\^o's formula with the Fr\'echet derivative with respect to the forward variance curve, while the Fr\'echet differentiability in an appropriate space was not verified.

This paper aims at a rigorous treatment of SSR based on the It\^o–Wentzell and the Clark–Ocone formulae, and applying it to analyze its  asymptotics under Bergomi‑type  models.
In Section~2, we derive a representation formula of SSR in terms of the Malliavin-Shigekawa derivative.
In Section~3, we apply the formula to derive the short-maturity limit of SSR under the Bergomi-type models.
In Section~4, we analyze the small volatility-of-volatility asymptotics.

\section{SSR formulae}
 Here we give a representation formula for the SSR. Let
$(\Omega,\mathscr{F},\mathsf{P},\{\mathscr{F}_t\})$ be
an filtered probability space. We assume that the filtration $\{\mathscr{F}_t\}$ is the augmentation of  one generated by an $n$ dimensional Brownian motion $(B^1,\dots,B^n)$.
Let $S = \{S_t\}$ be an asset price process given by
\begin{equation}\label{mod1}
     S_t = S_0 \exp\left\{ \int_0^t \sqrt{V_u}\, \mathrm{d}B^1_u - \frac{1}{2}\int_0^t V_u \, \mathrm{d}u \right\}, \ \ S_0 > 0,
\end{equation}
for a positive continuous adapted process $V= \{V_t\}$. We assume interest rates to be zero for brevity and that
the put option price at $t\geq 0$ with maturity $T>t$ and strike price $K>0$ is given by
\begin{equation*}
P_t(K) = \mathsf{E}[(K-S_T)_+|\mathscr{F}_t].
\end{equation*}
Here and here after, all the equalities and inequalities for random variables in this paper are in the almost sure sense.
We assume $P(K) = \{P_t(K)\}$  to be
    twice continuously differentiable in $K$ in the following sense:
\begin{hypothesis}\label{H1}
   There exist continuous processes $P^\prime(K) = \{P^\prime_t(K)\}$ and $P^{\prime\prime}(K) = \{P^{\prime\prime}_t(K)\}$ for each $K$ such that
    $P^\prime_t(K)$ and $P^{\prime\prime}_t(K)$ almost surely coincide  respectively with the first and second derivatives of $P_t(K)$ in $K$ for each $t \geq 0$, 
    and moreover
    $(t,K)\mapsto (P_t(K),P_t^\prime(K),P_t^{\prime\prime}(K))$ is continuous on $[0,T)\times (0,\infty)$.
\end{hypothesis}

By the dominated convergence theorem, the derivative $P^\prime_t(K)$ satisfies
\begin{equation}\label{pp}
P^\prime_t(K) = \mathsf{P}[K> S_T|\mathscr{F}_t]
\end{equation}
if exists. In particular, both $P(K)$ and $P^\prime(K)$ are bounded continuous martingales for each $K$ under Hypothesis~\ref{H1}.
Then, by the martingale representation theorem, there exist
progressively measurable processes $f^i(K)$ and $g^i(K)$ such that
\begin{eqnarray}
    P_t(K) &= P_0(K) + \sum_{i=1}^n \int_0^t f^i_s(K)\, \mathrm{d}B^i_s,\label{fi}\\ 
    P^\prime_t(K) &= P^\prime_0(K) + \sum_{i=1}^n \int_0^t g^i_s(K)\, \mathrm{d}B^i_s. \label{gi}
\end{eqnarray}
We require the following technical condition:
\begin{hypothesis}\label{H2}
For each $t < T$ and $i=1,\dots,n$, 
$f^i_t(K)$ and $g^i_t(K)$ in \eqref{fi} and \eqref{gi} are  continuous in $K$ with
\begin{equation*}
    \int_0^t \sup_{K \in \mathbb{K}}|f^i_s(K)|^2\, \mathrm{d}s
    + \int_0^t \sup_{K \in \mathbb{K}}|g^i_s(K)|^2\, \mathrm{d}s < \infty
\end{equation*}
for any compact set $\mathbb{K} \subset (0,\infty)$.
\end{hypothesis}
Here we give a version of the It\^o-Wentzell formula which plays a key role in this study:
\begin{lemma}\label{IW}
Under Hypotheses~\ref{H1} and \ref{H2}, for any $t <T$,
\begin{equation*}
    \begin{split}
        P_t(S_t) = P_0(S_0) &+ \int_0^t P^\prime_u(S_u)\,\mathrm{d}S_u + \frac{1}{2}\int_0^t P^{\prime\prime}_u(S_u)\,\mathrm{d}\langle S \rangle_u \\
        &+ \sum_{i=1}^n \int_0^t f^i_u(S_u)\, \mathrm{d}B^i_u
        + \sum_{i=1}^n \int_0^t g^i_u(S_u)\, \mathrm{d}\langle B^i,S\rangle_u.
    \end{split}
\end{equation*}
\end{lemma}
\begin{proof}
We are in a slightly different framework from the abstract one given in Chapter~3 of \cite{Kunita} but essentially the same proof works as follows:
    For any partition $\Delta: 0 = t_0 < t_1 < \dots < t_m = t$,
    \begin{equation*}
        \begin{split}
         P_t(S_t) - P_0(S_0) &= \sum_{k=0}^{m-1} (P_{t_{k+1}}(S_{t_{k+1}}) - P_{t_{k}}(S_{t_{k}})) \\
         &=
         \sum_{k=0}^{m-1} (P_{ t_{k+1}}(S_{t_{k}}) - P_{ t_{k}}(S_{t_{k}})) 
         + \sum_{k=0}^{m-1} (P_{ t_{k+1}}(S_{t_{k+1}}) - P_{ t_{k+1}}(S_{t_{k}})).
        \end{split}
    \end{equation*}
    The first sum is computed as
    \begin{equation*}
         \sum_{k=0}^{m-1} (P_{ t_{k+1}}(S_{t_{k}}) - P_{ t_{k}}(S_{t_{k}})) 
         = \sum_{k=0}^{m-1}\sum_{i=1}^n
         \int_{t_k}^{t_{k+1}}f^i_u(S_{t_k}) \, \mathrm{d}B^i_u
         \to \sum_{i=1}^n \int_0^t f^i_u(S_u)\, \mathrm{d}B^i_u
    \end{equation*}
    as $|\Delta| \to 0$. 
    The second sum is decomposed using the Taylor theorem as
    \begin{equation*}
        \begin{split}
            & \sum_{k=0}^{m-1} (P_{ t_{k+1}}(S_{t_{k+1}}) - P_{ t_{k+1}}(S_{t_{k}})) \\
            =\, &
             \sum_{k=0}^{m-1} (P_{t_{k+1}}^\prime(S_{t_k})-P_{t_k}^\prime(S_{t_k}))(S_{t_{k+1}}-S_{t_k}) 
             +   \sum_{k=0}^{m-1} P_{t_k}^\prime(S_{t_k})(S_{t_{k+1}}-S_{t_k}) \\
             &+ \frac{1}{2} \sum_{k=0}^{m-1} \int_0^1
             P_{t_{k+1}}^{\prime\prime}(S_{t_k} + \theta (S_{t_{k+1}}-S_{t_k}))\,
             \mathrm{d}\theta  (S_{t_{k+1}}-S_{t_k})^2.
        \end{split}
    \end{equation*}
    The first sum of this decomposition is computed as
    \begin{equation*}
    \begin{split}   
    \sum_{k=0}^{m-1} (P_{t_{k+1}}^\prime(S_{t_k})-P_{t_k}^\prime(S_{t_k}))(S_{t_{k+1}}-S_{t_k})&  = 
           \sum_{k=0}^{m-1} \sum_{i=1}^n \int_{t_k}^{t_{k+1}}g^i_u(S_{t_k}) \, \mathrm{d}B^i_u (S_{t_{k+1}}-S_{t_k})
        \\ & \to\sum_{i=1}^n \int_0^t g^i_u(S_u) \, \mathrm{d}\langle B^i, S \rangle_u
    \end{split} 
    \end{equation*}
    as $|\Delta| \to 0$.
    The second and third sums are computed respectively as
    \begin{equation*}
         \sum_{k=0}^{m-1} P_{t_k}^\prime(S_{t_k})(S_{t_{k+1}}-S_{t_k}) \to 
         \int_0^t P_u^\prime(S_u)\, \mathrm{d}S_u 
    \end{equation*}
    and
    \begin{equation*}
        \frac{1}{2} \sum_{k=0}^{m-1} \int_0^1
             P_{t_{k+1}}^{\prime\prime}(S_{t_k} + \theta (S_{t_{k+1}}-S_{t_k}))\,
             \mathrm{d}\theta  (S_{t_{k+1}}-S_{t_k})^2 \to \frac{1}{2}\int_0^t P^{\prime\prime}_u(S_u)\,\mathrm{d}\langle S \rangle_u
    \end{equation*}
as $|\Delta| \to 0$,  which concludes the proof.
\end{proof}

Let $\Sigma(K)$ be the implied total variance defined by \eqref{def1}
and define the at-the-money total variance $\Sigma^S = \{\Sigma^S_t\}$ by 
$\Sigma^S_t = \Sigma_t(S_t)$. 
\begin{proposition}\label{prop1}
Under Hypotheses~\ref{H1} and \ref{H2}, $\Sigma^S$ is a continuous semimartingale with
\begin{equation*}
    \frac{\mathrm{d}\langle \Sigma^S, S \rangle_t}
{ \mathrm{d}\langle S \rangle_t}
= 
 \frac{2\sqrt{\Sigma^S_t}}{S_t\phi\left(\frac{\sqrt{\Sigma^S_t}}{2}\right)} \left(
 \frac{\mathsf{E}[1_{\{S_t > S_T\}}S_T|\mathscr{F}_t]}{S_t} 
 + \frac{f^1_t(S_t)}{S_t \sqrt{V_t}} \right)
\end{equation*}
 for any $t <T$,
where $\phi$ is the standard normal density function.
\end{proposition}
\begin{proof}
    Let $P^S_t$ denote $P_t(S_t)$. Then, by Lemma~\ref{IW} and \eqref{pp}, $P^S$ is a continuous semimaringale with
    \begin{equation*}
       \mathrm{d} \langle P^S,S\rangle_t = P^\prime_t(S_t)\mathrm{d}\langle S \rangle_t +
       f^1_t(S_t)\mathrm{d}\langle B^1, S\rangle_t
       = \left( \mathsf{P}[S_t > S_T|\mathscr{F}_t] + \frac{f^1_t(S_t)}{S_t \sqrt{V_t}}\right)\mathrm{d}\langle S \rangle_t.
    \end{equation*}
    On the other hand, by \eqref{def1}, we have
    \begin{equation*}
        P^S_t = S_t \left(\Phi\left(\frac{\sqrt{\Sigma^S_t}}{2}\right)
 -\Phi\left(-\frac{\sqrt{\Sigma^S_t}}{2}\right)\right)
    \end{equation*}
    and so, $\Sigma^S$ is a continuous semimartingale and
    \begin{equation*}
        \mathrm{d}\langle P^S, S \rangle_t = \left(\Phi\left(\frac{\sqrt{\Sigma^S_t}}{2}\right)
 -\Phi\left(-\frac{\sqrt{\Sigma^S_t}}{2}\right)\right)\mathrm{d}\langle S \rangle_t +
 \frac{S_t\phi\left(\frac{\sqrt{\Sigma^S_t}}{2}\right)}{2\sqrt{\Sigma^S_t}} \mathrm{d}\langle \Sigma^S, S \rangle_t.
    \end{equation*}
    Combining these,  we obtain
\begin{equation*}
\frac{\mathrm{d}\langle \Sigma^S, S \rangle_t}
{ \mathrm{d}\langle S \rangle_t}
=
 \frac{2\sqrt{\Sigma^S_t}}{S_t\phi\left(\frac{\sqrt{\Sigma^S_t}}{2}\right)} 
 \left(\mathsf{P}[S_t\geq S_T|\mathscr{F}_t] -\Phi\left(\frac{\sqrt{\Sigma^S_t}}{2}\right)
 +\Phi\left(-\frac{\sqrt{\Sigma^S_t}}{2}\right)
 + \frac{f^1_t(S_t)}{S_t \sqrt{V_t}} \right).
\end{equation*}    
Notice that
\begin{equation*}
         \mathsf{E}[1_{\{K > S_T\}}S_T|\mathscr{F}_t]
 =  K\mathsf{P}[K > S_T|\mathscr{F}_t] -\mathsf{E}[(K - S_T)_+|\mathscr{F}_t]
\end{equation*}
and so
\begin{equation*}
\begin{split}
          \frac{\mathsf{E}[1_{\{S_t > S_T\}}S_T|\mathscr{F}_t]}{S_t}
 &=   \mathsf{P}[S_t > S_T|\mathscr{F}_t]-\frac{P_t(S_t)}{S_t}
 \\& = \mathsf{P}[S_t > S_T|\mathscr{F}_t] - \Phi\left(\frac{\sqrt{\Sigma^S_t}}{2}\right)+ \Phi\left(-\frac{\sqrt{\Sigma^S_t}}{2}\right),
\end{split}
   \end{equation*}
   which concludes the proof.
\end{proof}

Let the SSR $R_t$ be defined by \eqref{def2}. 
By a simple computation, we have an alternative representation
\begin{equation}\label{eqR}
    R_t = \frac{1}{\Sigma^\prime_t(S_t)}
    \frac{\mathrm{d}\langle \Sigma^S,S \rangle_t}{\mathrm{d}\langle S\rangle_t},
\end{equation}
which is more convenient in the following.
As is well-known, differentiating \eqref{def1} in $K$,
\begin{equation*}
\mathsf{P}[K  > S_T|\mathscr{F}_t] = 
    \Phi(-d_-(S_t,\Sigma_t(K))) +
    \frac{\partial p_K}{\partial { w}}(S_t,\Sigma_t(K))\Sigma^\prime_t(K),
\end{equation*}
which implies
  \begin{equation}
    \label{eqSk}
    \Sigma^\prime_t(S_t) =  \frac{2\sqrt{\Sigma^S_t}}{S_t\phi\left(\frac{\sqrt{\Sigma^S_t}}{2}\right)} \left(\mathsf{P}[S_t  > S_T|\mathscr{F}_t] - \Phi\left(\frac{\sqrt{\Sigma^S_t}}{2}\right)\right).
\end{equation}

\begin{theorem}\label{SSRformula}
 Under Hypotheses \ref{H1} and \ref{H2}, if
$\log S_T \in \mathbb{D}^{1,2}$ and $\Sigma^\prime_t(S_t) \neq 0$,
then $Y\neq 0$ and $R = X/Y$, where
\begin{equation}\label{SSR1}
\begin{split}
&X_t =
\mathsf{E}\left[1_{\{S_t > S_T\}} \frac{S_T}{S_t} \left(
1 - \frac{\mathcal{D}^1_t \log S_T}{\sqrt{V_t}} \right) 
\Bigg|\mathscr{F}_t\right],\\
& Y_t =      \mathsf{P}[S_t  > S_T|\mathscr{F}_t] - \Phi\left(\frac{\sqrt{\Sigma^S_t}}{2}\right),
\end{split}
\end{equation}
$(\mathcal{D}^1,\dots,\mathcal{D}^n)$ is the Malliavin-Shigekawa derivative operator with respect to the Brownian motion $(B^1,\dots, B^n)$, and
$\mathbb{D}^{1,2}$ is the $(1,2)$-Sobolev space with 
respect to $(\mathcal{D}^1,\dots,\mathcal{D}^n)$.
\end{theorem}
\begin{proof}
 By Proposition~\ref{prop1} and \eqref{eqSk}, we have
\begin{equation*}
    R_t = X_t/Y_t, \ \ X_t = \frac{\mathsf{E}[1_{\{S_t > S_T\}}S_T|\mathscr{F}_t]}{S_t} 
 + \frac{f^1_t(S_t)}{S_t \sqrt{V_t}}.
\end{equation*}
By the Clark-Ocone formula (see, e.g., \cite{NN} for the case $n=1$, for which the proof is extended to the general case in a straightforward manner),
\begin{equation}\label{CO}
\begin{split}
P_T(K) &=      (K-S_T)_+  \\
&= \mathsf{E}[(K-S_T)_+] +  \sum_{i=1}^n
\int_0^T  \mathsf{E}[\mathcal{D}^i_t(K-S_T)_+|\mathscr{F}_t]\mathrm{d}B^i_t\\
&=
\mathsf{E}[(K-S_T)_+] - \sum_{i=1}^n
\int_0^T  \mathsf{E}[
1_{\{ K> S_T\}}  S_T \mathcal{D}^i_t \log S_T
|\mathscr{F}_t]\mathrm{d}B^i_t.
\end{split}
\end{equation} 
See Exercise 3.3 of \cite{NN} for a generalized chain rule used here. 
Comparing with \eqref{fi}, we obtain the claimed expression of $f^1_t(S_t)$.
\end{proof}

\begin{theorem}\label{thm2}
    Assume a Bergomi-type model for $V_t(s) = \mathsf{E}[V_s|\mathscr{F}_t]$:
    \begin{equation}\label{BergomiM}
    \mathrm{d}V_t(s) = V_t(s)\sum_{i=1}^dk_i(s-t)\mathrm{d}W^i_t,\ \ t <s, 
\end{equation}
where  
$(W^1,\dots,W^d)$ is a $d$-dimensional standard Brownian motion correlated with $B^1$ in \eqref{mod1} as 
\begin{equation*}
    \mathrm{d}\langle B^1, W^i\rangle_t = \rho_i \mathrm{d}t,  \ \  \rho: = \sqrt{\sum_{i=1}^d \rho_i^2} \in [0,1),
\end{equation*}
 $k_i$, $i=1,\dots,d$, are locally square integrable functions on $[0,\infty)$, and 
   $t \mapsto V_0(t)$ is a deterministic positive continuous function.
Then, Hypotheses~\ref{H1} and \ref{H2} are satisfied, $\log S_T \in \mathbb{D}^{1,2}$,
and $X_t$ defined by \eqref{SSR1} satisfies
\begin{equation*}
X_t = -
    \frac{1}{2S_t \sqrt{V_t}}     \mathsf{E}\left[1_{\{S_t>S_T\}}S_T \int_t^T \sqrt{V_s}k(s-t)\left(\mathrm{d}B^1_s - \sqrt{V_s}\,\mathrm{d}s\right)\Bigg|\mathscr{F}_t\right],
\end{equation*}
where
\begin{equation}\label{defk}
      k = \sum_{i=1}^d \rho_i k_i.
\end{equation} 

\end{theorem}
\begin{proof} First, we note that the system \eqref{BergomiM} is solved as
\begin{equation*}
\begin{split}
      S_T &= S_t \exp\left\{\int_t^T \sqrt{V_s} \mathrm{d}B^1_s -\frac{1}{2}\int_t^T V_s \mathrm{d}s \right\}, \\ 
    V_s &= V_0(s) 
    \exp\left\{\sum_{i=1}^d \int_0^s k_i(s-u) \mathrm{d}W^i_u -\frac{1}{2}\sum_{i=1}^d \int_0^s k_i(s-u)^2 \mathrm{d}u\right\}.
\end{split}
\end{equation*}
Second, $I_d - (\rho_1,\dots,\rho_d)^\top(\rho_1,\dots,\rho_d)$ has the eigenvalues $1$ and $1-\rho^2$, and so is regular, with 
the symmetric Cholesky factor
\begin{equation*}
   I_d - (\rho_1,\dots,\rho_d)^\top(\rho_1,\dots,\rho_d)
   = LL^\top, \ \ L = 
    I_d - \beta (\rho_1,\dots,\rho_d)^\top(\rho_1,\dots,\rho_d),
\end{equation*}
where $\beta = (1 - \sqrt{1-\rho^2})/\rho^2$. This implies that $(B^1,\dots,B^{d+1})$ defined by
\begin{equation*}
   \begin{bmatrix}
   B_2 \\
   \vdots \\ B_{d+1} \end{bmatrix}
   =  L^{-1} \begin{bmatrix}
       W^1 - \rho_1 B^1
       \\
   \vdots \\
   W^d - \rho_d B^1
   \end{bmatrix}
\end{equation*}
is  a $d+1$-dimensional standard Brownian motion such that 
\begin{equation*}
   \begin{bmatrix}
       W^1 \\ \vdots \\  W^d 
   \end{bmatrix}
   =
   \begin{bmatrix}
       \rho_1 B^1 \\ \vdots \\ \rho_d B^1
   \end{bmatrix}
    + L \begin{bmatrix}
   B_2 \\
   \vdots \\ B_{d+1} \end{bmatrix}.
\end{equation*}
Therefore, the model can be formulated in the framework of Theorem~\ref{SSRformula}
 with $n=d+1$,
\begin{equation*}
    \mathcal{D}^1_t \sqrt{V_s} = \frac{\sqrt{V_s} }{2}\sum_{i=1}^d \rho_ik_i(s-t), \ \ 
      \mathcal{D}^1_t V_s =   V_s \sum_{i=1}^d \rho_ik_i(s-t)
      \end{equation*}
      for $s > t$ { and
      \begin{equation*}
            \mathcal{D}^j_t \sqrt{V_s} = \frac{\sqrt{V_s} }{2} \sum_{i=1}^d k_i(s-t)L_{ij},\ \ 
            \mathcal{D}^j_t V_s =V_s \sum_{i=1}^d k_i(s-t)L_{ij}
      \end{equation*}
      for $j= 2,\dots, n$ and $s > t$.} Since these are progressively measurable, 
\begin{equation*}
\begin{split}
     \mathcal{D}^1_t \log S_T & =  \sqrt{V_t} + \int_t^T \mathcal{D}^1_t\sqrt{V_s}\mathrm{d}B^1_s - \frac{1}{2}\int_t^T \mathcal{D}^1_tV_s\mathrm{d}s \\
     &=\sqrt{V_t} + \frac{1}{2}
     \left(\int_t^T \sqrt{V_s}k(s-t)\mathrm{d}B^1_s - \int_t^T V_sk(s-t)\mathrm{d}s \right)
\end{split}
\end{equation*}
by the standard computation. 
The claimed formula follows by substituting this to \eqref{SSR1}.

It remains to show that Hypotheses \ref{H1} and \ref{H2} are satisfied.
In light of \eqref{CO},
\begin{equation*}
    \begin{split}
        f^1_t(K) &=   \mathsf{E}[
1_{\{ K> S_T\}}  S_T \mathcal{D}^1_t \log S_T
|\mathscr{F}_t]\\
&=\sqrt{V_t}\mathsf{E}[
 1_{\{ K> S_T\}}  S_T 
|\mathscr{F}_t] \\
& \hspace*{.5cm}+
\frac{1}{2}
 \mathsf{E}\left[
1_{\{ K> S_T\}}  S_T 
 \left(\int_t^T \sqrt{V_s}k(s-t)\mathrm{d}B^1_s - \int_t^T V_sk(s-t)\mathrm{d}s \right)
|\mathscr{F}_t\right].
    \end{split}
\end{equation*}
Notice that
\begin{equation*}
    W^\perp := \frac{1}{\sqrt{1-\rho^2}}\left( B^1 - \sum_{i=1}^d \rho_iW^i\right)
\end{equation*}
is a standard Brownian motion independent of $(W^1,\dots,W^d)$, and $B^1$ is decomposed as
\begin{equation*}
    B^1 = \rho W + \sqrt{1-\rho^2} W^\perp, \ \ W := \frac{1}{\rho} \sum_{i=1}^d \rho_iW^i.
\end{equation*}
Using this decomposition, we have
\begin{equation*}
\begin{split}
     &S_T =^d S^\rho_t \exp\left\{ \sqrt{G_t}N -\frac{1}{2}G_t \right\}, \ \  G_t =  (1-\rho^2)\int_t^T V_u \mathrm{d}u,\\
    &S^\rho_t = S_t  \exp\left\{\rho\int_t^T \sqrt{V_u}\mathrm{d}W_u - \frac{1}{2}
    \rho^2\int_t^T V_u \mathrm{d}u \right\},
\end{split}
\end{equation*}
where $=^d$ denotes the equality in distribution and $N\sim \mathcal{N}(0,1)$ is independent of  the $\sigma$ algebra  $\mathcal{G}_T$ generated by 
$\mathscr{F}_t$ and
$\{(W^1_u,\dots,W^d_u)\}_{u \in [t,T]}$. Therefore,
\begin{equation*}
\begin{split}
  &   \mathsf{E}[
 1_{\{ K> S_T\}}  S_T 
|\mathscr{F}_t] \\&= 
\mathsf{E}[
\mathsf{E}[ 1_{\{ K> S_T\}}  S_T  | \mathscr{G}_T]
|\mathscr{F}_t]\\
&= \mathsf{E}\left[S^\rho_t \Phi\left( 
\frac{\log \frac{K}{S^\rho_t}}{\sqrt{G_t}} - \frac{\sqrt{G_t}}{2}
\right) | \mathscr{F}_t\right]
= \int_0^K \mathsf{E}\left[\frac{S^\rho_t}{k\sqrt{G_t}} \phi\left( 
\frac{\log \frac{k}{S^\rho_t}}{\sqrt{G_t}} - \frac{\sqrt{G_t}}{2}
\right) | \mathscr{F}_t\right] \mathrm{d}k.
\end{split}   
\end{equation*}
By similar computations, for each $i=1,\dots,n$ and $t <T$, we have
\begin{equation*}
    f^i_t(K) = \int_0^K \mathsf{E}\left[\psi^i\left( \log \frac{k}{S^\rho_t},G_t,H^i_t\right)|\mathscr{F}_t\right] \mathrm{d}k
\end{equation*}
for some smooth function $\psi^i$ and continuous (non-adapted) process $H^i = \{H^i_t\}$ satisfying
\begin{equation*}
    \int_0^t  \mathsf{E}\left[\sup_{a\in \mathbb{R}}\left| \frac{\partial^j}{\partial a^j } \psi^i\left(a,G_u,H^i_u\right) \right|^2|\mathscr{F}_u\right]\mathrm{d}u < \infty
\end{equation*}
for any $j=0,1,\dots$. This observation verifies Hypotheses \ref{H1} and \ref{H2} with
\begin{equation*}
    g^i_t(K) = \frac{\mathrm{d}}{\mathrm{d}K} f^i_t(K) = \mathsf{E}\left[\psi^i\left(\log\frac{K}{S^\rho_t},G_t,H^i_t\right)|\mathscr{F}_t\right],
\end{equation*}
which concludes the proof.
\end{proof}

\begin{remark}\upshape
The case $d=2$ with $k_i$ being exponential functions
$k_i(t) = a_ie^{-b_it}$, $a_i, b_i >0$ in \eqref{BergomiM},
describes the two factor Bergomi model (see Bergomi~\cite{Bergomi}).
The case $d=1$ with $k_1(t) = a t^{H-1/2}$, $a>0$, $H \in (0,1/2)$, describes the rough Bergomi model (see \cite{Rough}).
\end{remark}
\section{Short-maturity asymptotics}
In this section, we assume the Bergomi-type model \eqref{BergomiM}  and the conditions for Theorem~\ref{thm2} to hold.
We derive the limit of $R_t$ as $T \to t$.

Since
the filtration $\{\mathscr{F}_t\}$ is generated by a Brownian motion,
there exists a regular conditional probability measure $\mathsf{P}_t$ given $\mathscr{F}_t$.
Let $\mathsf{E}_t$ denote the expectation under $\mathsf{P}_t$ and define
\begin{equation*}
\begin{split}
X_t(T) &= -
    \frac{1}{2S_t \sqrt{V_t}}     \mathsf{E}_t\left[1_{\{S_t>S_T\}}S_T \int_t^T \sqrt{V_s}k(s-t)\left(\mathrm{d}B^1_s - \sqrt{V_s}\,\mathrm{d}s\right)\right],\\
    Y_t(T) &=      \mathsf{P}_t[S_t  > S_T] - \Phi\left(\frac{\sqrt{\Sigma^S_t(T)}}{2}\right),
    \end{split}
\end{equation*}
where $\Sigma_t^S(T)$ is defined through
\begin{equation*}
    \mathsf{E}_t[(S_t-S_T)_+] = p_{ S_t} (S_t,\Sigma^S_t(T)).
\end{equation*}
By Theorem~\ref{thm2}, $R$ coincides with $R(T):=X(T)/Y(T)$ almost everywhere with respect to $\mathrm{d}\mathsf{P} \otimes \mathrm{d}t$.
\begin{theorem}
Let $k$ defined by \eqref{defk}. If 
    there exists $H \in (0,1/2]$ such that the map $$u \mapsto g(u):= u^{1/2-H}k(u)$$ is continuous on $(0,\infty)$ with finite limit $g(0+) \neq 0$, then, 
    \begin{equation}\label{ssrs}
        \lim_{T\to t} R_t(T) = H + \frac{3}{2}.
    \end{equation}
    \end{theorem}
    \begin{proof}
    The regular conditional distribution of
    \begin{equation*}
      \frac{1}{\sqrt{V_t}}  \left(
(T-t)^{-1/2}\log \frac{S_T}{S_t},\, (T-t)^{-H}\int_t^T \sqrt{V_s}k(s-t)\left(\mathrm{d}B^1_s - \sqrt{V_s}\mathrm{d}s \right) \right) 
    \end{equation*}
    given $\mathscr{F}_t$ is uniformly integrable, and by the martingale central limit theorem { (see e.g., Theorem VIII.3.11 of \cite{JS})},
   it  converges in law 
    to a centered normal random variable $(Z_1,Z_2)$ with covariance
    \begin{equation*}
        \mathsf{E}[Z_1^2] = 1, \ \  \mathsf{E}[Z_1Z_2] = \frac{2g(0+)}{2H+1}, \ \ \mathsf{E}[Z_2^2] =  \frac{g(0+)^2}{2H}.
    \end{equation*}
Since $\mathsf{E}[Z_2|Z_1=z] = \mathsf{E}[Z_1Z_2]z$, we have then
\begin{equation*}
    (T-t)^{-H}X_t(T) \to -\frac{1}{2} \int_{-\infty}^{0}\mathsf{E}[Z_2|Z_1=z] \phi(z)\mathrm{d}z
    = \frac{g(0+)}{2H+1}\phi(0).
\end{equation*}
On the other hand, by Theorem~2 of \cite{TVP} (more precisely, by the proof of it), we know that
\begin{equation*}
    (T-t)^{-H}Y_t(T) \to  \frac{g(0+)}{2(H+1/2)(H+3/2)}\phi(0),
\end{equation*}
which concludes the proof. \end{proof}

\begin{remark}
Fukasawa~\cite{F21} gave a formal asymptotic analysis towards \eqref{ssrs}
by replacing 
$\langle \sigma^S,\log S \rangle$ with the covariance between spot volatility increments and log returns. 
Friz and Gatheral~\cite{FG} and
Bourgeys et al.~\cite{BDD,BDD2} gave numerical evidences of \eqref{ssrs}.
\end{remark}

\section{Small volatility-of-volatility asymptotics}
Here we introduce a small volatility-of-volatility parameter $\epsilon>0$
by replacing $k_i$ by $\epsilon k_i$ in \eqref{BergomiM}.
We have then
\begin{equation*} 
\begin{split}
S_t &= S^\epsilon_t := S_0 \exp\left\{\int_0^t \sqrt{V^\epsilon_s}\mathrm{d}B^1_s - \frac{1}{2}\int_0^t V^\epsilon_s \mathrm{d}s \right\},\\
     V_t &= V^\epsilon_t := V_0(t)\exp\left\{
\epsilon \sum_{i=1}^d \int_0^t k_i(t-s)\mathrm{d}W^i_s - \frac{\epsilon^2}{2}
\sum_{i=1}^d \int_0^t k_i(t-s)^2 \mathrm{d}s
    \right\}.
\end{split}
\end{equation*}
Here we consider the limit of $R_t$ as $\epsilon \to 0$.
For simplicity, we let $t=0$.  
Let $k$ be defined by \eqref{defk} and define
\begin{equation*}
\begin{split}
X(\epsilon) &= -
    \frac{\epsilon}{2S_0 \sqrt{V_0}}     \mathsf{E}\left[1_{\{S_0>S_T\}}S_T \int_0^T \sqrt{V_s}k(s)\left(\mathrm{d}B^1_s - \sqrt{V_s}\,\mathrm{d}s\right)\right],\\
    Y(\epsilon) &=      \mathsf{P}[S_0  > S_T] - \Phi\left(\frac{\sqrt{\Sigma^S_0}}{2}\right),
    \end{split}
\end{equation*}
and $R(\epsilon) = X(\epsilon)/Y(\epsilon)$ in light of Theorem~\ref{thm2}.
 
\begin{theorem}\label{main}
If
\begin{equation*}
    \int_0^T \sqrt{V_0(s)}\int_s^T V_0(u)k(u-s)\mathrm{d}u \mathrm{d}s\neq 0,
\end{equation*}
then
    \begin{equation}\label{ssre}
     \lim_{\epsilon \to 0} R(\epsilon) = \frac{\int_0^TV_0(s)\mathrm{d}s \int_0^TV_0(u)k(u)\mathrm{d}u}{\sqrt{V_0}\int_0^T \sqrt{V_0(s)}\int_s^T V_0(u)k(u-s)\mathrm{d}u \mathrm{d}s}.
        \end{equation}
\end{theorem}
\begin{proof}
The $2$-dimensional random vector
 \begin{equation*}
      \left(
\log \frac{S^\epsilon_T}{S_0},\, \int_0^T \sqrt{V^\epsilon_s}k(s)\left(\mathrm{d}B^1_s - \sqrt{V^\epsilon_s}\mathrm{d}s \right) \right) 
    \end{equation*}
    is uniformly integrable, and by the martingale central limit theorem  { (see e.g., Theorem VIII.3.11 of \cite{JS})},
    it converges in law to a $2$-dimensional normal random vector $(Z_1,Z_2)$ with mean vector, covariance matrix,
    \begin{equation*}
        \left( - \frac{A}{2}, - B\right), \ \ 
        \ \begin{pmatrix}
           A &
        B
         \\ 
          B &
       C
       \end{pmatrix}
    \end{equation*}
respectively, where
    \begin{equation*}
        A = \int_0^T V_0(s) \mathrm{d}s, \ \ B = \int_0^T V_0(s)k(s)\mathrm{d}s, \ \ C =   \int_0^T V_0(s)k(s)^2 \mathrm{d}s.
    \end{equation*}
    Therefore,
    \begin{equation*}
    \begin{split}
        \frac{X(\epsilon)}{\epsilon} \to 
       -  \frac{1}{2\sqrt{V_0}}    
         \mathsf{E}\left[1_{\{Z_1 < 0\}}e^{Z_1}Z_2\right]
       =
         \frac{B}{2\sqrt{V_0 A}}\phi\left(-\frac{\sqrt{A}}{2}\right).
    \end{split}
    \end{equation*}
    On the other hand, from Section~3 of \cite{ME}, we know that
    \begin{equation*}
        \frac{Y(\epsilon)}{\epsilon} \to
        \frac{1}{2 A^{3/2}}\phi\left(-\frac{\sqrt{A}}{2}\right)
        \int_0^T \sqrt{V_0(s)}\int_s^T V_0(u)k(u-s)\mathrm{d}u \mathrm{d}s,
    \end{equation*}
    which yields the claim.
\end{proof}

\begin{remark}\upshape
 When the initial forward variance curve $V_0(t)$ is flat and $k= at^{H-1/2}$, then
        the right hand side of \eqref{ssre} is $H+3/2$, close to the empirical estimate $3/2$ by Bergomi~\cite{Bergomi} for the Euro Stoxx and S\&P markets when $H\approx 0$.
The asymptotics \eqref{ssre} has been formally derived by Bergomi~\cite{Bergomi}
for his $n$-factor model and by Bourgeys et al.~\cite{BDD2} for rough volatility models.
\end{remark}

\end{document}